\documentclass[%
 reprint, amsmath,amssymb,superscriptaddress,
 aps]{revtex4-1}
\usepackage{enumerate}
\usepackage{amsmath}
\usepackage{booktabs}
\usepackage{array}
\usepackage{amsthm}
\usepackage{float}
\usepackage{dcolumn}
\usepackage[colorlinks=true, citecolor=blue, urlcolor=blue ]{hyperref}
\usepackage{bm}
\usepackage{longtable}
\usepackage{graphicx}
\usepackage{epsfig}
\usepackage{epstopdf}
\usepackage{amssymb}
\usepackage{color}
\newcommand{\ucite}[1]{\textsuperscript{\cite{#1}}}

\begin{document}

\preprint{APS/123-QED}

\title{Locally stable sets with minimum cardinality}

	\author{Hai-Qing Cao}
\affiliation{School of Mathematical Sciences, Hebei Normal University, Shijiazhuang, 050024, China}
	\author{Mao-Sheng Li}
\affiliation{ School of Mathematics,
	South China University of Technology, Guangzhou
	510641,  China}

\author{Hui-Juan Zuo}
\email{huijuanzuo@163.com}
\affiliation{School of Mathematical Sciences, Hebei Normal University, Shijiazhuang, 050024, China}

\affiliation{ Hebei Key Laboratory of Computational Mathematics and Applications, Shijiazhuang, 050024, China}
\affiliation{Hebei International Joint Research Center for Mathematics and Interdisciplinary Science, Shijiazhuang, 050024, China
}%


\begin{abstract}
The nonlocal set has received wide attention over recent years. Shortly before, Li and Wang \href{https://arxiv.org/abs/2202.09034} {arXiv: 2202.09034} proposed the concept of a locally stable set: the only possible orthogonality preserving measurement on each subsystem is trivial. Locally stable sets present  stronger nonlocality than those sets that are just locally indistinguishable. In this work, we focus on the constructions of locally stable sets in multipartite quantum systems. First, two lemmas are put forward to prove that an orthogonality-preserving local measurement must be trivial. Then we present the constructions of locally stable sets with minimum cardinality in bipartite quantum systems $\mathbb{C}^{d}\otimes \mathbb{C}^{d}$ $(d\geq 3)$ and $\mathbb{C}^{d_{1}}\otimes \mathbb{C}^{d_{2}}$ $(3\leq d_{1}\leq d_{2})$. Moreover, for the multipartite quantum systems $(\mathbb{C}^{d})^{\otimes n}$ $(d\geq 2)$ and $\otimes^{n}_{i=1}\mathbb{C}^{d_{i}}$ $(3\leq d_{1}\leq d_{2}\leq\cdots\leq d_{n})$, we also obtain $d+1$ and $d_{n}+1$ locally stable orthogonal states respectively. Fortunately, our constructions reach the lower bound of the cardinality on the locally stable sets, which provides a positive and complete answer to an open problem raised in \href{https://arxiv.org/abs/2202.09034}{arXiv: 2202.09034}.

\begin{description}
\item[PACS numbers]
03.65.Ud, 03.67.Mn
\end{description}
\end{abstract}

\pacs{Valid PACS appear here}
\maketitle


\section{\label{sec:level1}Introduction\protect}

A set of orthogonal quantum states is locally indistinguishable if it is not possible to optimally distinguish the states by any sequence of local operations and classical communications (LOCC). In 1999, Bennett {\it et al}. \cite{Bennett1999} first presented a set of locally indistinguishable orthogonal product bases in $\mathbb{C}^{3}\otimes \mathbb{C}^{3}$, which shows the phenomenon of nonlocality without entanglement. With the increasing research of nonlocality, there are many relevant references on locally indistinguishable orthogonal entangled states \cite{Walgate2000,Ghosh2001,Walgate2002,Ghosh2004,Fan2004,Nathanson2005,Cohen2007,Yu2012,Yu2015,Li2015,Wang2019}.
Especially, the locally indistinguishable orthogonal product states have attracted more attention \cite{Niset2006,Yang2013,Zhangzc2014,Wangyl2015,Zhangzc2015,Zhangxq2016,Xugb2016,Zhangzc2016,Xu2016,Wangyl2017,Zhangzc2017,Halder2018,Jiang2020,Xu2021,Zuo2022,Zhen2022} . Its closely related research branch, entanglement-assisted discrimination protocol, has also achieved fruitful results \cite{Zhangzc2018,Li2019,Lilj2019}. The local indistinguishability has wide applications in quantum cryptographic protocols such as secret sharing and data hiding \cite{Yang2015,Wang2017,Jiangdh2020,GuoGP2001,Rahaman2015,DiVincenzo2002}. That is the reason why so many scholars are engaged in the research of local discrimination of quantum states.

In 2019, Rout {\it et al}. \cite{Rout2019} proposed the concept of genuine nonlocality based on local indistinguishability. Then many interesting results spring up like mushrooms \cite{Rout2019,Li2021,Rout2021}. Recently, Halder {\it et al}. \cite{Halder2019} put forward the concept of strong nonlocality based on locally irreducible quantum states. A set of multipartite orthogonal product states is strongly nonlocal if it is locally irreducible in every bipartition. Many people began to engage in the research and a few
strongly nonlocal sets were obtained \cite{Halder2019,Zhangzc2019,Yuan2020,Shi2020,Shi2022,Li2022}.

An important method was provided to verify the local indistinguishability of orthogonal product states in Ref. \cite{Walgate2002}, which showed the fact that no matter which party goes first, he (or she) can only perform a trivial measurement. Since 2014, a great deal of research (see Refs. \cite{Zhangzc2014,Wangyl2015,Zhangzc2015,Zhangxq2016,Xugb2016,Zhangzc2016,Xu2016,Wangyl2017,Zhangzc2017,Halder2018,Jiang2020,Xu2021,Zuo2022}) on the locally indistinguishable sets of quantum states is based on the aforementioned observation. Li {\it et al}. \cite{Li2022} concentrated on the orthogonal sets of multipartite quantum states with the property: the only possible orthogonality preserving measurement on each subsystem is trivial. The set with such property is called a locally stable set. Note that locally stable sets are always locally indistinguishable. Hence they could be used  to show some particular form of distinguishability-based  nonlocality.   Li {\it et al}. \cite{Li2022} also obtained a lower bound of the cardinality on the locally stable set, i.e., if $\mathcal{S}$ is a locally stable set of orthogonal pure states in $\mathcal{H}:=\otimes^{N}_{i=1}\mathcal{H}_{A_{i}}$,
whose local dimension is dim$_{\mathbb{C}}$($\mathcal{H}_{A_{i}}$)=$d_{i}$, then~$\left| \mathcal{S} \right |$ $\geq$ max$_{i}\{d_{i} + 1\}$.  They conjectured that    this lower bound may be tight.  That is,  there may exist some locally stable set $\mathcal{S}\subseteq \mathcal{H}$  whose cardinality is exactly the aforementioned lower bound  $\max_{i}\{d_{i} + 1\}.$  In this work, we will provide a positive answer to this conjecture.
  Adding any orthogonal states to a locally stable set  (nonlocal set)  forms a new set which is again locally stable (nonlocal). Hence it is interesting to find the optimal locally stable set in the sense that,  removing any state from this set, it is impossible to achieve local stability again. Therefore, those locally stable sets in $\mathcal{H}$ with cardinality being  $\max_{i}\{d_{i} + 1\}$ are   always optimal.

In the manuscript, we aim to construct locally stable sets whose cardinality reach the lower bound indicated in Ref. \cite{Li2022}  for general multipartite quantum systems. Fortunately, we prove that there exist $d+1$ orthogonal states in $\mathbb{C}^{d}\otimes \mathbb{C}^{d}$ $(d\geq 3)$ and $d_{2}+1$ orthogonal states in $\mathbb{C}^{d_{1}}\otimes \mathbb{C}^{d_{2}}$ $(3\leq d_{1}\leq d_{2})$ are locally stable. For the multipartite cases, we present two constructions of locally stable sets in multipartite quantum systems $(\mathbb{C}^{d})^{\otimes n}$ $(d\geq 2)$ and $\otimes^{n}_{i=1}\mathbb{C}^{d_{i}}$ $(3\leq d_{1}\leq d_{2}\leq\cdots\leq d_{n})$, which contain $d+1$ and $d_{n}+1$ orthogonal states, respectively. All of the locally stable sets can reach the minimum cardinality on the locally stable set proposed in Ref. \cite{Li2022}. In addition, we found another structure of the smallest locally stable set in $\otimes^{n}_{i=1}\mathbb{C}^{d_{i}}$ $(3\leq d_{1}\leq d_{2}\leq\cdots\leq d_{n})$, which is composed of genuine entangled states apart from one full product state.

\theoremstyle{remark}
\newtheorem{definition}{\indent Definition}
\newtheorem{lemma}{\indent Lemma}
\newtheorem{theorem}{\indent Theorem}
\newtheorem{proposition}{\indent Proposition}
\newtheorem{corollary}{\indent Corollary}
\def\QEDclosed{\mbox{\rule[0pt]{1.3ex}{1.3ex}}}|
\def\QED{\QEDclosed}
\def\proof{\noindent{\indent\em Proof}.}
\def\endproof{\hspace*{\fill}~\QED\par\endtrivlist\unskip}

\section{Preliminaries}

Throughout this paper, we only consider pure states and we do not normalize states for simplicity. Here we take the computational basis $\{|i\rangle\}_{i=0}^{d_{k}-1}$ for each $d_{k}$-dimensional subsystem. For simplicity, we denote the state $\frac{1}{\sqrt{n}}(|i_1\rangle \pm |i_2\rangle\pm \cdots \pm |i_n\rangle)$ as $|i_1\pm i_2\pm\cdots \pm i_n\rangle$, $\mathbb Z_{d_{k}}=\{0, 1,\cdots, d_{k}-1\}$, $(\mathbb{C}^{d})^{\otimes n}=\mathbb{C}^{d}\otimes\mathbb{C}^{d}\otimes \cdots\otimes\mathbb{C}^{d}$, and $\otimes^{n}_{i=1}\mathbb{C}^{d_{i}}=\mathbb{C}^{d_{1}}\otimes\mathbb{C}^{d_{2}}\otimes \cdots\otimes\mathbb{C}^{d_{n}}$.
In particular, it should be pointed out that the stopper state has the expression \begin{equation}\label{eq:stopper}
|S\rangle=\otimes^{n}_{k=1}(\sum\limits_{i_{k} \in \mathbb{Z}_{d_{k}}}|i_{k}\rangle_{A_{k}})
\end{equation} For each integer $d\geq 2$, we denote $w_d=e^{\frac{2\pi \sqrt{-1}}{d}},$ i.e., a  primitive  $d$th root of unit.

All the participants perform   positive operator-valued measures (POVM) on their local sites.  Each $k$th subsystem's POVM element $M_{k}^{\dagger}M_{k}$ can be represented by a $d_{k}\times d_{k}$ matrix $E_k=(m^k_{a,b})_{a,b \in \mathbb{Z}_{d_{k}}}$ in the computational basis.  A POVM  is called a trivial measurement if all its  elements   are proportional to the identity operator. To ensure the local distinguishability the postmeasurement states should remain orthogonal. We observe that in each locally distinguishable protocol, each local measurement must preserve the orthogonality of the states. Using this observation, there is a widely used method for deducing the local indistinguishability of an orthogonal set: to preserve the orthogonality of the states, each party could only perform trivial measurement.  This method motivates the definition of locally stable.

\begin{definition}(Locally indistinguishable)\ucite{Bennett1999}
A set of orthogonal pure states in multipartite quantum systems is said to be locally indistinguishable, if it is not possible to distinguish the states by using LOCC.
\end{definition}

\begin{definition}(Locally irreducible)\ucite{Halder2019}
A set of orthogonal quantum states on $\mathcal{H}=\otimes^{n}_{i=1}\mathcal{H}_{i}$ with $n\geq 2$ and dim$\mathcal{H}_{i}\geq 2$, $i=1,2,\cdots,n$ is locally irreducible if it is not possible to eliminate one or more states from the set by orthogonality-preserving local measurements.
\end{definition}

\begin{definition}(Locally stable)\ucite{Li2022}\label{def:stable}
An orthogonal set of pure states in multipartite quantum systems is said to be locally stable if the only possible orthogonality preserving measurement on the subsystems is trivial.
\end{definition}

In Ref. \cite{Li2022}, it is shown that locally stable sets are always locally irreducible and locally irreducible sets are always locally indistinguishable; the converse is not true. Therefore, locally stable sets present the strongest form of quantum nonlocality among the three classes: locally indistinguishable sets, locally irreducible sets, and locally stable sets.

 Given an orthogonal set $\mathcal{S}=\{|\phi_i\rangle\}_{i=1}^N$ of pure states in  $\otimes^{n}_{i=1}\mathbb{C}^{d_{i}}$, if the $k$th party starts with the first orthogonality preserving measurement whose measurement element is denoted as   $E_k=(m^k_{a,b})_{a,b \in \mathbb{Z}_{d_{k}}},$  then we have
\begin{equation}\label{eq:orthogonal}
	 \langle \phi_i | I_1\otimes I_2\otimes \cdots \otimes E_k \otimes \cdots \otimes I_n |\phi_j\rangle =0
	 \end{equation}
  for all different pairs $|\phi_i\rangle,|\phi_j\rangle \in \mathcal{S}.$
Now  we put forward two simple lemmas which are useful for deducing an orthogonality preserving measurement $E_k=(m^k_{a,b})_{a,b \in \mathbb{Z}_{d_{k}}}$ to be a trivial one.


\begin{lemma} (Zero entries)\label{lem:zero}
Fix $k\in\{1,2,\cdots,n\}$. Suppose that
$$
\begin{array}{rcl}
	|\phi_{i}\rangle&=&\sum\limits^{p_{i}-1}_{t=0}\omega_{p_{i}}^{t}|i_{1}^{t}\rangle_{A_{1}}|i_{2}^{t}\rangle_{A_{2}}\cdots|i_{n}^{t}\rangle_{A_{n}},\\
|\phi_{j}\rangle&=&\sum\limits^{p_{j}-1}_{s=0}\omega_{p_{j}}^{s}|j_{1}^{s}\rangle_{A_{1}}|j_{2}^{s}\rangle_{A_{2}}\cdots|j_{n}^{s}\rangle_{A_{n}},
\end{array}$$  where $|i_{1}^{t}\rangle_{A_{1}}|i_{2}^{t}\rangle_{A_{2}}\cdots|i_{n}^{t}\rangle_{A_{n}}$ and $|j_{1}^{s}\rangle_{A_{1}}|j_{2}^{s}\rangle_{A_{2}}\cdots|j_{n}^{s}\rangle_{A_{n}}$ are mutually orthogonal and there is only one pair $(t_0,s_0)\in \mathbb{Z}_{p_i} \times \mathbb{Z}_{p_j}$ such that
$$ \prod_{\ell\neq k} \langle i_\ell^{t_0}| j_\ell^{s_0}\rangle_{A_\ell} \neq 0.$$ Then the equation  $\langle \phi_i | I_1\otimes I_2\otimes \cdots \otimes E_k \otimes \cdots \otimes I_n |\phi_j\rangle =0$  implies  that $m^k_{i^{t_{0}}_{k},j^{s_{0}}_{k}}=0$.
\end{lemma}

\begin{lemma}(Diagonal entries)\label{lem:Dia}
Fix $k\in\{1,2,\cdots,n\}$. Let $|S\rangle$ be the stopper state defined in Eq.~\eqref{eq:stopper} and   $|\phi_{i}\rangle=\sum\limits ^{p-1}_{t=0}\omega_{p}^{t}|i_{1}^{t}\rangle _{A_{1}}|i_{2}^{t}\rangle_{A_{2}}\cdots|i_{n}^{t}\rangle_{A_{n}}$  where   there exist only two different values among $i^{0}_{k}, i^{1}_{k},\cdots, i^{p-1}_{k}$, say,  $i^{t_{0}}_{k}$ and $i^{t_{1}}_{k}$.  If   all the  off-diagonal entries of the matrix $E_k=(m^k_{a,b})_{a,b \in \mathbb{Z}_{d_{k}}}$  are zeros, then  the equation  $\langle S| I_1\otimes I_2\otimes \cdots \otimes E_k \otimes \cdots \otimes I_n |\phi_i\rangle =0$ implies that $m^k_{i^{t_{0}}_{k},i^{t_{0}}_{k}}=m^k_{i^{t_{1}}_{k},i^{t_{1}}_{k}}$.
\end{lemma}

The proofs of the above two Lemmas are given in Appendix \ref{append:A}.

\section{Constructions in bipartite quantum systems}
In this section, we propose the construction of locally stable sets with minimum cardinality in $\mathbb{C}^{d}\otimes \mathbb{C}^{d}$~$(d\geq 3)$ and $\mathbb{C}^{d_{1}}\otimes \mathbb{C}^{d_{2}}$~$(3\leq d_{1}\leq d_{2})$.

	\begin{figure}[h]
	\centering
	\includegraphics[scale=0.75]{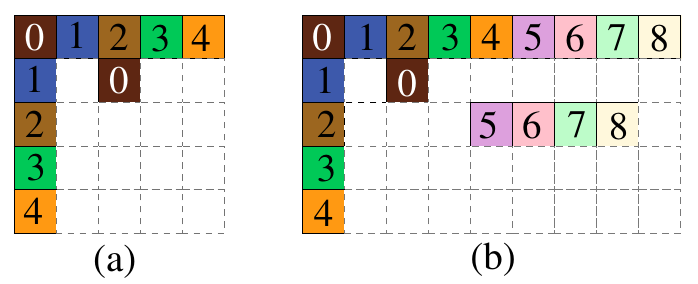}
	\caption{Intuition of the structure of states we constructed in Eq.~\eqref{eq:bidd} and  Eq.~\eqref{eq:bid1d2}. Figure (a) corresponds to the systems $\mathbb{C}^5\otimes \mathbb{C}^5$, while (b) corresponds to systems   $\mathbb{C}^5\otimes \mathbb{C}^9$. The squares indicated by the same color represent a unique state and the numbers represent the subscripts of the states. For example, the two orange squares (4,0) and (0,4) with label ``4" correspond to the state $|\phi_{4}\rangle=|40\rangle_{AB}-|04\rangle_{AB}$; the two pink squares (0,6) and (2,5) with label ``6" in the right side correspond to the state $|\phi_{6}\rangle=|06\rangle_{AB}-|25\rangle_{AB}$. }\label{Fig:bipartite}
\end{figure}
\subsection{locally stable set in $\mathbb{C}^d\otimes \mathbb{C}^d$}
\begin{theorem}\label{the:$d+1$} The following set $\mathcal{S}$ of $d+1$ orthogonal states is locally stable in $\mathbb{C}^d\otimes \mathbb{C}^d$ [see Fig. \ref{Fig:bipartite}(a) for an intuition of the example where  $d=5$]:
\begin{equation}\label{eq:bidd}
\begin{aligned}
|\phi_{0}\rangle & =|00\rangle_{AB}-|12\rangle_{AB}, \\
|\phi_{i}\rangle&=|i0\rangle_{AB}-|0i\rangle_{AB},\\
|S\rangle & =|0+\cdots+(d-1)\rangle_{A}|0+\cdots+(d-1)\rangle_{B},
\end{aligned}
\end{equation}
where $i=1,2,\cdots, d-1$,~$d\geq 3$.
\end{theorem}

\begin{proof} First, we assume that  Alice starts with the first measurement. Let $E_1=(m^1_{a,b})_{a,b \in \mathbb{Z}_d}$ represent  an element of any orthogonality-preserving measurement performed by Alice. For  each pair $|\psi\rangle,|\phi\rangle\in \mathcal{S}$  with $|\psi\rangle\neq |\phi\rangle$, we have
\begin{equation}\label{eq:ddA}
	 \langle \psi| E_1\otimes I_2|\phi\rangle=0.
\end{equation}

For $1\leq i\neq j\leq d-1$, considering Eq.~\eqref{eq:ddA} for the states $|\phi_{i}\rangle$ and $|\phi_{j}\rangle$,   we obtain $m^1_{i,j}= m^1_{j,i}= 0$ directly from  Lemma~\ref{lem:zero}.   Now we  consider  Eq.~\eqref{eq:ddA} for the states   $|\phi_{0}\rangle$ and $|\phi_{i}\rangle$ for $i=1,2,\cdots,d-1.$ If $i\in \{1,3,\cdots, d-1\}$,   we can get $m^1_{0,i}= m^1_{i,0}= 0$  from  Lemma~\ref{lem:zero}.
If $i=2$, the corresponding   Eq.~\eqref{eq:ddA} is just $(\langle 00|-\langle 12|)E_{1}\otimes I_{2}(|20\rangle -|02\rangle)=0$, i.e., $\langle 0|E_{1}|2\rangle \langle 0|I_{2}|0\rangle-\langle 0|E_{1}|0\rangle \langle0|I_{2}|2\rangle-\langle 1|E_{1}|2\rangle\langle 2|I_{2}|0\rangle+\langle 1|E_{1}|0\rangle \langle 2|I_{2}|2\rangle=0$, which gives rise to $m^1_{0,2}+m^1_{1,0}=0$. Since $m^1_{0,1}=m^1_{1,0}=0$, we can get $m^1_{0,2}=0$. Thus $m^1_{0,i}=m^1_{i,0}=0$ for $1\leq i\leq d-1$. Therefore, the off-diagonal entries of  $E_1$ are all zeros.

For $1\leq i \leq d-1$, considering Eq.~\eqref{eq:ddA} for the states $|S\rangle$ and $|\phi_{i}\rangle$, we get $m^1_{0,0}= m^1_{i,i}$   by Lemma~\ref{lem:Dia}. Therefore, $E_{1} $ is proportional to the identity matrix. Hence Alice can only start with a trivial measurement.

Suppose that Bob starts with the first orthogonality-preserving measurement whose elements are  represented as $E_2=(m^2_{a,b})_{a,b \in \mathbb{Z}_d}.$  Then for  each pair $|\psi\rangle,|\phi\rangle\in \mathcal{S}$  with $|\psi\rangle\neq |\phi\rangle$, we have
\begin{equation}\label{eq:ddB}
	\langle \psi| I_1\otimes E_2|\phi\rangle=0.
\end{equation}

 In the same way, considering Eq.~\eqref{eq:ddB} for the states $|\phi_{i}\rangle$ and $|\phi_{j}\rangle$,   we obtain $m^2_{i,j}= m^2_{j,i}= 0$ directly from  Lemma~\ref{lem:zero}  for $1\leq i\neq j\leq d-1$.
  Now we consider Eq.~\eqref{eq:ddB}  for the states $|\phi_{0}\rangle$ and $|\phi_{i}\rangle$ for $1\leq i\leq d-1$. If  $2\leq i\leq d-1$, we have $m^2_{0,i}= m^2_{i,0}= 0$  by Lemma~\ref{lem:zero}.
If  $i=1$, we have  $(\langle 00|-\langle 12|)I_{1}\otimes E_{2}(|10\rangle -|01\rangle)=0$, i.e., $\langle 0|I_{1}|1\rangle \langle 0|E_{2}|0\rangle-\langle 0|I_{1}|0\rangle \langle0|E_{2}|1\rangle-\langle 1|I_{1}|1\rangle \langle 2|E_{2}|0\rangle+\langle 1|I_{1}|0\rangle \langle 2|E_{2}|1\rangle=0$, which gives rise to $m^2_{0,1}+m^2_{2,0}=0$. Since $m^2_{0,2}=m^2_{2,0}=0$, we can get $m^2_{0,1}=0$. Thus $m^2_{0,i}=m^2_{i,0}=0$ for $1\leq i\leq d-1$.  Therefore, the off-diagonal entries of  $E_2$ are all zeros.

For $1\leq i \leq d-1$, considering  Eq.~\eqref{eq:ddB} for the states $|S\rangle$ and $|\phi_{i}\rangle$, we get $m^2_{0,0}= m^2_{i,i}$   by Lemma~\ref{lem:Dia}.
Therefore, $E_{2}$ is proportional to the identity matrix. Bob can only implement a trivial orthogonality-preserving measurement also.

Thus the above $d+1$ states are locally stable by definition. This completes the proof.
\end{proof}

Specifically, the construction is not unique, where $|\phi_{0}\rangle =|00\rangle_{AB}-|12\rangle_{AB}$ can be $|\phi_{0}\rangle =|00\rangle_{AB}-|1k\rangle_{AB}$ $(2\leq k\leq d-1)$. This is true for other examples presented.
\subsection{locally stable set in $\mathbb{C}^{d_{1}}\otimes \mathbb{C}^{d_{2}}$}

\begin{theorem}\label{the:$d_{2}+1$} Let $3\leq d_{1}\leq d_{2}$. The following set $\mathcal{S}$ of $d_{2}+1$ orthogonal states is locally stable in $\mathbb{C}^{d_{1}}\otimes \mathbb{C}^{d_{2}}$  [see Fig. \ref{Fig:bipartite} (b) for an intuition of the example where $d_1=5$ and $d_2=9$]:
\begin{equation}\label{eq:bid1d2}
\begin{aligned}
|\phi_{0}\rangle & =|00\rangle_{AB}-|12\rangle_{AB}, \\
|\phi_{i}\rangle&=|i0\rangle_{AB}-|0i\rangle_{AB},~~~1\leq i\leq d_{1}-1\\
|\phi_{j}\rangle&=|0j\rangle_{AB}-|2(j-1)\rangle_{AB},~~~d_{1}\leq j\leq d_{2}-1, \\
|S\rangle & =|0+\cdots+(d_{1}-1)\rangle_{A}|0+\cdots+(d_{2}-1)\rangle_{B}.
\end{aligned}
\end{equation}

\end{theorem}

\begin{proof} Obviously, Alice could only start with trivial orthogonality-preserving measurement by the same argument as  case  $\mathbb{C}^{d}\otimes \mathbb{C}^{d}$. We only need to show  that the orthogonality-preserving measurement Bob could perform is the  trivial one. Suppose that Bob starts with the first orthogonality-preserving measurement whose elements are  represented as $E_2=(m^2_{a,b})_{a,b\in\mathbb{Z}_{d_2}}.$  Then for  each pair $|\psi\rangle,|\phi\rangle\in \mathcal{S}$  with $|\psi\rangle\neq |\phi\rangle$, we have
	\begin{equation}\label{eq:d1d2B}
		\langle \psi| I_1\otimes E_2|\phi\rangle=0.
	\end{equation}

With a similar argument as the case $\mathbb{C}^d\otimes \mathbb{C}^d,$ we could obtain that $m^2_{i,i'}=m^2_{i',i}=0$ for all  $0\leq i\neq i'\leq d_1-1.$
Considering Eq.~\eqref{eq:d1d2B} for the states $|\phi_{0}\rangle$ and $|\phi_{j}\rangle$,  we directly get $m^2_{0,j}= m^2_{j,0}= 0$ for $d_{1}\leq j\leq d_{2}-1$ by Lemma~\ref{lem:zero}.

Now we consider Eq.~\eqref{eq:d1d2B} for  the states $|\phi_{i}\rangle$ and $|\phi_{j}\rangle$  for $1\leq i\leq d_1-1,$ and $d_1\leq j\leq d_2-1.$  If $i\neq 2$,   we get $m^2_{i,j}=m^2_{j,i}=0$  directly from  Lemma~\ref{lem:zero}. If $i=2$, we have $\langle 2|I_{1}|0\rangle \langle 0|E_2|j\rangle-\langle 2|I_{1}|2\rangle \langle0|E_2|j-1\rangle-\langle 0|I_{1}|0\rangle \langle 2|E_2|j\rangle+\langle 0|I_{1}|2\rangle \langle 2|E_2|j-1\rangle=0$ which deduces that $m^2_{0,j-1}+m^2_{2,j}=0$. Since $m^2_{0,j-1}=m^2_{j-1,0}=0$, we  have $m^2_{2,j}=0.$  Therefore, we have $m^2_{i,j}=m^2_{j,i}=0$ for all $1\leq i\leq d_{1}-1, d_{1}\leq j\leq d_{2}-1$.

Then  we consider Eq.~\eqref{eq:d1d2B} for  the states $|\phi_{j}\rangle$ and $|\phi_{j'}\rangle$  for $d_1\leq j<j'\leq d_2-1.$ That is, we have the equation  $\langle 0|I_{A}|0\rangle \langle j|E_2|j'\rangle-\langle 0|I_{1}|2\rangle \langle j|E_2|j'-1\rangle-\langle 2|I_{1}|0\rangle \langle j-1|E_2|j'\rangle+\langle 2|I_{1}|2\rangle \langle j-1|E_2|j'-1\rangle=0$, which implies that $m^2_{j,j'}=-m^2_{(j-1),(j'-1)}.$  Therefore, $$m^2_{j,j'}=-m^2_{(j-1),(j'-1)}= \cdots=(-1)^{j-d_{1}+1}m^2_{d_{1}-1,(j'-j+d_{1}-1)}$$
which is equal to zero as the last term $m^2_{d_{1}-1,(j'-j+d_{1}-1)}=0$ has been obtained.
Thus we get $m^2_{j',j}=m^2_{j,j'}=0$ for $d_{1}\leq j<j'\leq d_{2}-1$. Up to now, we have shown that the off-diagonal entries of  $E_2$ are all zeros.

For $1\leq i \leq d_1-1$, considering Eq.~\eqref{eq:ddB} for the states  $|S\rangle$ and $|\phi_{i}\rangle$, we   get $m^2_{i,i}=m^2_{0,0}$   by Lemma~\ref{lem:Dia}. Similarly, considering the states $|S\rangle$ and $|\phi_{j}\rangle$ for $d_{1}\leq j\leq d_{2}-1$, we have $m^2_{j,j}= m^2_{j-1,j-1}$ which implies that $$m^2_{j,j}= m^2_{j-1,j-1}=\cdots=m^2_{d_1-1,d_1-1}=m^2_{0,0}.$$
Therefore, $E_{2}\propto \mathbb {I}$. Bob cannot start with a nontrivial measurement either.

In summary, both participants can only start with a trivial orthogonality-preserving measurement. Thus the above $d_{2}+1$ states are locally stable. This completes the proof.
\end{proof}

\section{Constructions In Multipartite Quantum Systems}
In this section, we put forward the constructions of the locally stable sets in multipartite quantum systems $(\mathbb{C}^{d})^{\otimes n}$ $(d\geq 2, n\geq 3)$ and $\otimes^{n}_{i=1}\mathbb{C}^{d_{i}}$ $(3\leq d_{1}\leq d_{2}\leq\cdots\leq d_{n},~n\geq 3)$.
\subsection{locally stable set in $(\mathbb{C}^{d})^{\otimes n}$}

\begin{theorem}\label{the:$d+1$}  In $(\mathbb{C}^{d})^{\otimes n}$ $(d\geq 2, n\geq 3)$, the following set $\mathcal{S}$ of $d+1$ orthogonal states are locally stable [see Fig. \ref{Fig:triddd} for an intuition of the example where $n=3$ and $d=5$]:
\begin{widetext}
\begin{equation}\label{eq:samed}
\begin{aligned}
|\phi_{0}\rangle & =|00\cdots00\rangle_{A_{1}A_{2}\cdots A_{n}}-|11\cdots11\rangle_{A_{1}A_{2}\cdots A_{n}}, \\
|\phi_{i}\rangle&=|i0\cdots00\rangle_{A_{1}A_{2}\cdots A_{n}}+\omega_{n}|0i\cdots00\rangle_{A_{1}A_{2}\cdots A_{n}}+\cdots+\omega_{n}^{n-1}|00\cdots0i\rangle_{A_{1}A_{2}\cdots A_{n}},\\
|S\rangle & =|0+\cdots+(d-1)\rangle_{A_{1}}|0+\cdots+(d-1)\rangle_{A_{2}}\cdots|0+\cdots+(d-1)\rangle_{A_{n}},
\end{aligned}
\end{equation}
\end{widetext}
where $1\leq i\leq d-1$, $d\geq 2$.
\end{theorem}
	
\begin{figure}[h]
	\centering
	\includegraphics[scale=0.64]{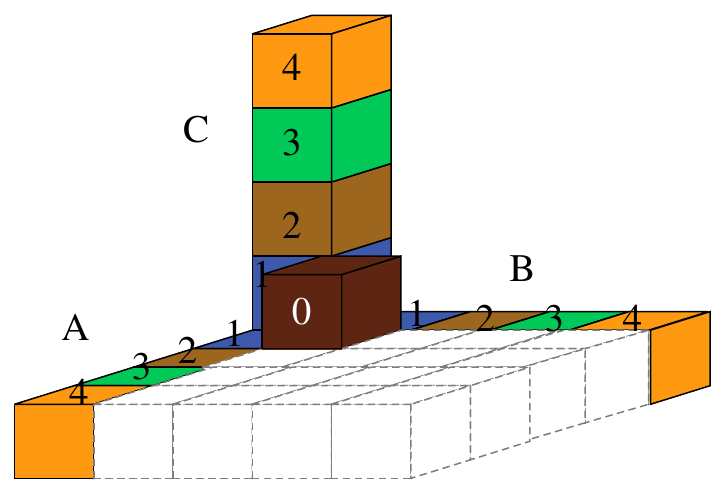}
	\caption{Intuition of the structure of states we constructed in Eq.~\eqref{eq:samed} for the setting $n=3$ and $d=5.$  Note that the cubic with coordinate $(0,0,0)$ should be labeled with $``0"$. The squares indicated by the same color represent a unique state and the numbers represent the subscripts of the states. For example, the orange squares (4,0,0), (0,4,0), and (0,0,4) with label $``4"$ correspond to the state $|\phi_{4}\rangle=|400\rangle_{ABC}+\omega_{3}|040\rangle_{ABC}+\omega_{3}^{2}|004\rangle_{ABC}$.}\label{Fig:triddd}
\end{figure}

\begin{proof} Since the  states are symmetric, it is sufficient  to prove that  the first party  could only start with a trivial    orthogonality-preserving measurement. Let $E_1=(m^1_{a,b})_{a,b\in\mathbb{Z}_d}$ represent  an element of any orthogonality-preserving measurement performed by Alice. For  each pair $|\psi\rangle,|\phi\rangle\in \mathcal{S}$  with $|\psi\rangle\neq |\phi\rangle$, we have
	\begin{equation}\label{eq:ddddA}
		\langle \psi| E_1\otimes I_2\otimes \cdots \otimes I_n|\phi\rangle=0.
	\end{equation}

For $0\leq i\neq j\leq d-1, $ considering Eq.~\eqref{eq:ddddA} for the states $|\phi_{i}\rangle$ and $|\phi_{j}\rangle$, we obtain $m^1_{i,j}= m^1_{j,i}= 0$ by Lemma \ref{lem:zero}. Therefore, the off-diagonal elements of  $E_1$  are all zeros. Considering Eq.~\eqref{eq:ddddA} for   the states $|S\rangle$ and $|\phi_{i}\rangle$, we directly get $m^1_{i,i}= m^1_{0,0}$ for $1\leq i\leq d-1$ by Lemma~\ref{lem:Dia}.
Therefore, $E_{1} $ is proportional to the identity matrix. So, the first party cannot start with a nontrivial orthogonality-preserving measurement.

 Therefore, the above $d+1$ states are locally stable by definition. This completes the proof.
\end{proof}

Next, we consider the constructions of locally stable sets in the general multipartite quantum systems. In order to be better understood, we first show our construction in arbitrary tripartite quantum systems.
\subsection{locally stable set in $\mathbb{C}^{d_{1}}\otimes \mathbb{C}^{d_{2}}\otimes \mathbb{C}^{d_{3}}$}

\begin{theorem}\label{the:d1d2d3} In $\mathbb{C}^{d_{1}}\otimes \mathbb{C}^{d_{2}}\otimes \mathbb{C}^{d_{3}}$ ($3\leq d_{1}\leq d_{2}\leq d_{3}$), the following set of $d_{3}+1$ orthogonal states is locally stable  [see Fig. \ref{Fig:trid1d2d3}   for an intuition of the example where $d_1=5,d_2=7$ and $d_3=10$]:
\begin{widetext}
\begin{equation}\label{eq:trid1d2d3}
\begin{aligned}
|\phi_{0}\rangle & =|000\rangle_{ABC}-|111\rangle_{ABC}, \\
|\phi_{i}\rangle&=|i00\rangle_{ABC}+\omega_{3}|0i0\rangle_{ABC}+\omega_{3}^{2}|00i\rangle_{ABC},~~~1\leq i\leq d_{1}-1\\
|\phi_{j}\rangle&=|0j0\rangle_{ABC}-|00j\rangle_{ABC},~~~d_{1}\leq j\leq d_{2}-1\\
|\phi_{k}\rangle&=|00k\rangle_{ABC}-|21(k-1)\rangle_{ABC},~~~d_{2}\leq k\leq d_{3}-1, \\
|S\rangle & =|0+\cdots+(d_{1}-1)\rangle_{A}|0+\cdots+(d_{2}-1)\rangle_{B}|0+\cdots+(d_{3}-1)\rangle_{C}.\\
\end{aligned}
\end{equation}
\end{widetext}

\end{theorem}

	\begin{figure}[h]
	\centering
	\includegraphics[scale=0.64]{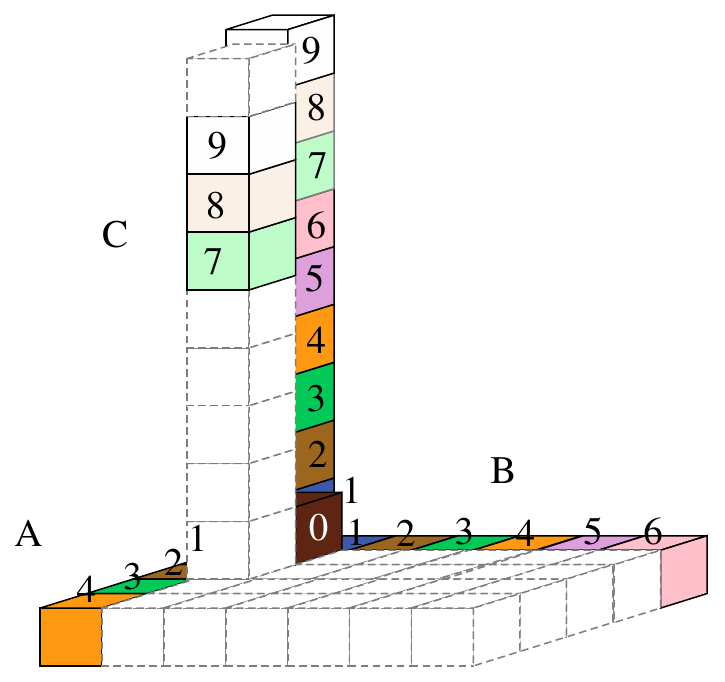}
		\caption{Intuition of the structure of states we constructed in Eq.~\eqref{eq:trid1d2d3} for the setting $d_1=5,d_2=7,d_3=10.$  Note that the cubic with coordinate $(0,0,0)$ should be labeled with $``0"$. The squares indicated by the same color represent a unique state and the numbers represent the subscripts of the states. For example, the two pink squares (0,6,0) and (0,0,6) with label $``6"$ correspond to the state $|\phi_{6}\rangle=|060\rangle_{ABC}-|006\rangle_{ABC}$; the two white squares (0,0,9) and (2,1,8) with label $``9"$ correspond to the state $|\phi_{9}\rangle=|009\rangle_{ABC}-|218\rangle_{ABC}$.}\label{Fig:trid1d2d3}
\end{figure}
\begin{proof} As far as Alice is concerned, it is the same as the equal dimensional case. We only need to prove that Bob and Charlie have to implement trivial measurement.

 Suppose that Bob starts with the first orthogonality-preserving measurement whose elements are  represented as $E_2=(m^2_{a,b})_{a,b\in\mathbb{Z}_{d_2}}.$  Then for each pair $|\psi\rangle,|\phi\rangle\in \mathcal{S}$  with $|\psi\rangle\neq |\phi\rangle$, we have
\begin{equation}\label{eq:d1d2d3B}
	\langle \psi| I_1\otimes E_2\otimes I_3|\phi\rangle=0.
\end{equation}

 Now we consider Eq.~\eqref{eq:d1d2d3B} for  the states $|\phi_{i}\rangle$ and $|\phi_{j}\rangle$  for $0\leq i\neq j\leq d_2-1.$
Note that the $AC$ parties of each term of  $|\phi_{i}\rangle$ and $|\phi_{j}\rangle$ are orthogonal except the terms $|0i0\rangle$ (corresponding to $|\psi_i\rangle$) and  $|0j0\rangle$ (corresponding to $|\psi_j\rangle$).
Therefore, by Lemma \ref{lem:zero}, we could obtain that  $m^2_{i,j}=m^2_{j,i}=0.$  Therefore,   the off-diagonal entries of  $E_2$ are all zeros.


For $1\leq i \leq d_2-1$, considering  Eq.~\eqref{eq:d1d2d3B} for the states $|S\rangle$ and $|\phi_{i}\rangle$, we get $ m^2_{i,i}=m^2_{0,0}$   by Lemma~\ref{lem:Dia}.
Therefore, $E_{2}\propto \mathbb {I}$. Bob cannot start with a nontrivial measurement either.

Let us consider the third party Charlie.
 Suppose that Charlie starts with the first orthogonality-preserving measurement whose elements are  represented as $E_3=(m^3_{a,b})_{a,b\in\mathbb{Z}_{d_3}}.$  Then for  each pair $|\psi\rangle,|\phi\rangle\in \mathcal{S}$  with $|\psi\rangle\neq |\phi\rangle$, we have
 \begin{equation}\label{eq:d1d2d3C}
 	\langle \psi| I_1\otimes I_2\otimes E_3|\phi\rangle=0.
 \end{equation}
Considering Eq.~\eqref{eq:d1d2d3C}  for the pair $|\phi_i\rangle$ and $|\phi_j\rangle$ where $0\leq i\neq j\leq d_3-1$, by  Lemma~\ref{lem:zero}, we could obtain that the off diagonal  entry $m^3_{i,j}$ of the matrix $E_{3}$ is zero except $d_2\leq i\neq j\leq d_3-1$ (see  Table \ref{3_result}).

\begin{table}[htbp]
	\newcommand{\tabincell}[2]{\begin{tabular}{@{}#1@{}}#2\end{tabular}}
	\centering
	\caption{\label{3_result}Zero entries of the matrix $E_3=(m^3_{a,b})_{a,b\in\mathbb{Z}_{d_3}}$.}
	\begin{tabular}{ccc}
		\toprule
		\hline
		\hline
		\specialrule{0em}{1.5pt}{1.5pt}
		~A pair of states~~&~~~Zero entries~~~&~~~Value range~~~\\
		\specialrule{0em}{1.5pt}{1.5pt}
		\midrule
		\hline
		\specialrule{0em}{1.5pt}{1.5pt}
		\tabincell{c}{$|\phi_{0}\rangle$,~$|\phi_{i}\rangle$}&\tabincell{c}{$m^3_{0,i}= m^3_{i,0}= 0$}&\tabincell{c}{$1\leq i\leq d_{1}-1$} \\
		\specialrule{0em}{3.5pt}{3.5pt}
		\tabincell{c}{$|\phi_{0}\rangle$,~$|\phi_{j}\rangle$}&\tabincell{c}{$m^3_{0,j}= m^3_{j,0}= 0$}&\tabincell{c}{$d_{1}\leq j\leq d_{2}-1$} \\
		\specialrule{0em}{3.5pt}{3.5pt}
		\tabincell{c}{$|\phi_{0}\rangle$,~$|\phi_{k}\rangle$}&\tabincell{c}{$m^3_{0,k}= m^3_{k,0}= 0$}&\tabincell{c}{$d_{2}\leq k\leq d_{3}-1$} \\
		\specialrule{0em}{3.5pt}{3.5pt}
		\tabincell{c}{$|\phi_{i}\rangle$,~$|\phi_{i'}\rangle$}&\tabincell{c}{$m^3_{i,i'}= m^3_{i',i}= 0$}&~~\tabincell{c}{$1\leq i\neq i'\leq d_{1}-1$} \\
		\specialrule{0em}{3.5pt}{3.5pt}
		\tabincell{c}{$|\phi_{j}\rangle$,~$|\phi_{j'}\rangle$}&\tabincell{c}{$m^3_{j,j'}= m^3_{j',j}= 0$}&~~\tabincell{c}{$d_{1}\leq j\neq j'\leq d_{2}-1$} \\
		\specialrule{0em}{3.5pt}{3.5pt}
		\tabincell{c}{$|\phi_{i}\rangle$,~$|\phi_{j}\rangle$}&\tabincell{c}{$m^3_{i,j}= m^3_{j,i}= 0$}&\tabincell{c}{$1\leq i \leq d_{1}-1$,\\$d_{1}\leq j \leq d_{2}-1$} \\
		\specialrule{0em}{3.5pt}{3.5pt}
		\tabincell{c}{$|\phi_{i}\rangle$,~$|\phi_{k}\rangle$}&\tabincell{c}{$m^3_{i,k}= m^3_{k,i}= 0$}&\tabincell{c}{$1\leq i \leq d_{1}-1$\\$d_{2}\leq k \leq d_{3}-1$} \\
		\specialrule{0em}{3.5pt}{3.5pt}
		\tabincell{c}{$|\phi_{j}\rangle$,~$|\phi_{k}\rangle$}&\tabincell{c}{$m^3_{j,k}= m^3_{k,j}= 0$}&\tabincell{c}{$d_{1}\leq j \leq d_{2}-1$\\$d_{2}\leq k \leq d_{3}-1$} \\
		\bottomrule
		\specialrule{0em}{1.5pt}{1.5pt}
		\hline
		\hline
	\end{tabular}
\end{table}

Now we consider  Eq.~\eqref{eq:d1d2d3C}  for the pair $|\phi_k\rangle$ and $|\phi_{k'}\rangle$ where $d_2\leq k<k' \leq d_3-1. $  That is,  $$[\langle 00k|-\langle 21(k-1)|]I_{1}\otimes I_{2}\otimes E_{3}[|00k'\rangle -|21(k'-1)\rangle]=0$$
from which we deduce that
  $m^3_{k,k'}=-m^3_{(k-1),(k'-1)}.$
  Therefore, we have
 $$m^3_{k,k'}=-m^3_{(k-1),(k'-1)}= \cdots=(-1)^{k-d_{2}+1}m^3_{d_{2}-1,(k'-k+d_{2}-1)},$$
 which is equal to zero as the last term $m^3_{d_{2}-1,(k'-k+d_{2}-1)}=0$ has been obtained.
Thus we get $m^3_{k,k'}=m^3_{k',k}=0.$

\begin{table}[htbp]
	\newcommand{\tabincell}[2]{\begin{tabular}{@{}#1@{}}#2\end{tabular}}
	\centering
	\caption{\label{4_result}~Diagonal entries of $E_3=(m^3_{a,b})_{a,b\in\mathbb{Z}_{d_3}}$.}
	\begin{tabular}{ccc}
		\toprule
		\hline
		\hline
		\specialrule{0em}{1.5pt}{1.5pt}
		A pair of states~~&~~Diagonal entries~&~Value Range\\
		\specialrule{0em}{1.5pt}{1.5pt}
		\midrule
		\hline
		\specialrule{0em}{1.5pt}{1.5pt}
		\tabincell{c}{$|S\rangle$,~$|\phi_{i}\rangle$}&\tabincell{c}{$m^3_{0,0}= m^3_{i,i}$}&~~\tabincell{c}{$1\leq i\leq d_{1}-1$} \\
		\specialrule{0em}{3.5pt}{3.5pt}
		\tabincell{c}{$|S\rangle$,~$|\phi_{j}\rangle$}&\tabincell{c}{$m^3_{0,0}= m^3_{j,j}$}&~~\tabincell{c}{$d_{1}\leq j \leq d_{2}-1$} \\
		\specialrule{0em}{3.5pt}{3.5pt}
		\tabincell{c}{$|S\rangle$,~$|\phi_{k}\rangle$}&\tabincell{c}{$m^3_{k,k}= m^3_{(k-1),(k-1)}$}&~~\tabincell{c}{$d_{2}\leq k \leq d_{3}-1$} \\
		\bottomrule
		\specialrule{0em}{1.5pt}{1.5pt}
		\hline
		\hline
	\end{tabular}
\end{table}

By Lemma~\ref{lem:Dia}, all diagonal entries of the matrix $E_{3}$ are equal from Table \ref{4_result}, i.e., $E_{3}\propto \mathbb {I}$. Charlie cannot start with a nontrivial measurement.

In summary, all the subsystems can only start with a trivial orthogonality-preserving measurement. Therefore, the above $d_{3}+1$ states form a locally stable set. This completes the proof.
\end{proof}

\subsection{locally stable set in $\otimes^{n}_{i=1}\mathbb{C}^{d_{i}}$}

\begin{table*}
	\newcommand{\tabincell}[2]{\begin{tabular}{@{}#1@{}}#2\end{tabular}}
	\centering
	\caption{\label{5_result} Zero entries of the matrix $E_{n}=(m_{a,b}^{n})_{a,b\in \mathbb Z_{d_{n}}}$.}
	\begin{tabular}{cccc}
		\toprule
		\hline
		\hline
		\specialrule{0em}{1.5pt}{1.5pt}
		~~~~~~~~~~~A pair of states ~~~~~~~~~~~~~&~~~~~~~~~~~~~~Zero entries~~~~~~~~~~~~~~~~~~&~~~~~~~~Value Range ~~~~~~~\\
		\specialrule{0em}{1.5pt}{1.5pt}
		\midrule
		\hline
		\specialrule{0em}{1.5pt}{1.5pt}
		\tabincell{c}{$|\phi_{0}\rangle$,~$|\phi_{i_{1}}\rangle$}&\tabincell{c}{$m_{0,i_{1}}^{n}= m_{i_{1},0}^{n}= 0$}&\tabincell{c}{$1\leq i_{1}\leq d_{1}-1$  } \\
		
		\tabincell{c}{\vdots}&\tabincell{c}{\vdots}&\tabincell{c}{\vdots}\\
		
		\tabincell{c}{$|\phi_{0}\rangle$,~$|\phi_{i_{n-1}}\rangle$}&\tabincell{c}{$m_{0,i_{n-1}}^{n}= m_{i_{n-1},0}^{n}= 0$}&\tabincell{c}{$d_{n-2}\leq i_{n-1}\leq d_{n-1}-1$  } \\
		\specialrule{0em}{3.5pt}{3.5pt}
		\tabincell{c}{$|\phi_{0}\rangle$,~$|\phi_{i_{n}}\rangle$}&\tabincell{c}{$m_{0,i_{n}}^{n}= m_{i_{n},0}^{n}= 0$}&\tabincell{c}{$d_{n-1}\leq i_{n}\leq d_{n}-1$} \\
		\specialrule{0em}{1.5pt}{1.5pt}
		\midrule
		\hline
		\specialrule{0em}{3.5pt}{3.5pt}
		\tabincell{c}{$|\phi_{i_{1}}\rangle$,~$|\phi_{i_{1}}'\rangle$}&\tabincell{c}{$m_{i_{1},i_{1}^{'}}^{n}= m_{i_{1}^{'},i_{1}}^{n}= 0$}&\tabincell{c}{$1\leq i_{1}\neq i_{1}^{'}\leq d_{1}-1$  } \\
		
		\tabincell{c}{\vdots}&\tabincell{c}{\vdots}&\tabincell{c}{\vdots} \\
		
		\tabincell{c}{$|\phi_{i_{n-2}}\rangle$,~$|\phi_{i_{n-2}}'\rangle$}&\tabincell{c}{$m_{i_{n-2},i_{n-2}^{'}}^{n}= m_{i_{n-2}^{'},i_{n-2}}^{n}= 0$}&\tabincell{c}{$d_{n-3}\leq i_{n-2}\neq i_{n-2}^{'}\leq d_{n-2}-1$  }\\
		\specialrule{0em}{3.5pt}{3.5pt}
		\tabincell{c}{$|\phi_{i_{n-1}}\rangle$,~$|\phi_{i_{n-1}}'\rangle$}&\tabincell{c}{$m_{i_{n-1},i_{n-1}^{'}}^{n}= m_{i_{n-1}^{'},i_{n-1}}^{n}= 0$}&\tabincell{c}{$d_{n-2}\leq i_{n-1}\neq i_{n-1}^{'}\leq d_{n-1}-1$ }\\
		\specialrule{0em}{1.5pt}{1.5pt}
		\midrule
		\hline
		\specialrule{0em}{3.5pt}{3.5pt}
		\tabincell{c}{$|\phi_{i_{1}}\rangle$,~$|\phi_{i_{2}}\rangle$}&\tabincell{c}{$m_{i_{1},i_{2}}^{n}= m_{i_{2},i_{1}}^{n}= 0$}&\tabincell{c}{$1\leq i_{1}\leq d_{1}-1$,~$d_{1}\leq i_{2}\leq d_{2}-1$  } \\
		
		\tabincell{c}{\vdots}&\tabincell{c}{\vdots}&\tabincell{c}{\vdots}\\
		
		\tabincell{c}{$|\phi_{i_{1}}\rangle$,~$|\phi_{i_{n-1}}\rangle$}&\tabincell{c}{$m_{i_{1},i_{n-1}}^{n}= m_{i_{n-1},i_{1}}^{n}= 0$}&\tabincell{c}{$1\leq i_{1}\leq d_{1}-1$,~$d_{n-2}\leq i_{n-1}\leq d_{n-1}-1$  } \\
		\specialrule{0em}{3.5pt}{3.5pt}
		\tabincell{c}{$|\phi_{i_{1}}\rangle$,~$|\phi_{i_{n}}\rangle$}&\tabincell{c}{$m_{i_{1},i_{n}}^{n}= m_{i_{n},i_{1}}^{n}= 0$}&\tabincell{c}{$1\leq i_{1}\leq d_{1}-1,$ $d_{n-1}\leq i_{n}\leq d_{n}-1$} \\
		\specialrule{0em}{3.5pt}{3.5pt}
		\tabincell{c}{$|\phi_{i_{2}}\rangle$,~$|\phi_{i_{3}}\rangle$}&\tabincell{c}{$m_{i_{2},i_{3}}^{n}= m_{i_{3},i_{2}}^{n}= 0$}&\tabincell{c}{$d_{1}\leq i_{2}\leq d_{2}-1$,~$d_{2}\leq i_{3}\leq d_{3}-1$  } \\
		
		\tabincell{c}{\vdots}&\tabincell{c}{\vdots}&\tabincell{c}{\vdots} \\
		
		\tabincell{c}{$|\phi_{i_{2}}\rangle$,~$|\phi_{i_{n-1}}\rangle$}&\tabincell{c}{$m_{i_{2},i_{n-1}}^{n}= m_{i_{n-1},i_{2}}^{n}= 0$}&\tabincell{c}{$d_{1}\leq i_{2}\leq d_{2}-1$,~$d_{n-2}\leq i_{n-1}\leq d_{n-1}-1$ } \\
		\specialrule{0em}{3.5pt}{3.5pt}
		\tabincell{c}{$|\phi_{i_{2}}\rangle$,~$|\phi_{i_{n}}\rangle$}&\tabincell{c}{$m_{i_{2},i_{n}}^{n}= m_{i_{n},i_{2}}^{n}= 0$}&\tabincell{c}{$d_{1}\leq i_{2}\leq d_{2}-1,$ $d_{n-1}\leq i_{n}\leq d_{n}-1$} \\
		
		\tabincell{c}{\vdots}&\tabincell{c}{\vdots}&\tabincell{c}{\vdots}\\
		
		\tabincell{c}{$|\phi_{i_{n-1}}\rangle$,~$|\phi_{i_{n}}\rangle$}&\tabincell{c}{$m_{i_{n-1},i_{n}}^{n}= m_{i_{n},i_{n-1}}^{n}= 0$}&\tabincell{c}{$d_{n-2}\leq i_{n-1}\leq d_{n-1}-1,$ $d_{n-1}\leq i_{n}\leq d_{n}-1$} \\
		\bottomrule
		\specialrule{0em}{1.5pt}{1.5pt}
		\hline
		\hline
	\end{tabular}
\end{table*}

\begin{theorem}
\label{the:dn} The following set $\mathcal{S}$ of $d_{n}+1$ orthogonal states are locally stable in $\otimes^{n}_{i=1}\mathbb{C}^{d_{i}}$ for $3\leq d_{1}\leq d_{2}\leq\cdots\leq d_{n}$ and ~$n\geq 3$:
\begin{widetext}
\begin{equation}\label{eq:mulpro}
\begin{aligned}
|\phi_{0}\rangle &=|00\cdots00\rangle_{A_{1}A_{2}\cdots A_{n}}-|11\cdots11\rangle_{A_{1}A_{2}\cdots A_{n}}, \\
|\phi_{i_{1}}\rangle&=|i_{1}0\cdots00\rangle_{A_{1}A_{2}\cdots A_{n}}+\omega_{n}|0i_{1}0\cdots00\rangle_{A_{1}A_{2}\cdots A_{n}}+\cdots+\omega_{n}^{n-2}|00\cdots0i_{1}0\rangle_{A_{1}A_{2}\cdots A_{n}}\\
&~~~~+\omega_{n}^{n-1}|00\cdots0i_{1}\rangle_{A_{1}A_{2}\cdots A_{n}},~1\leq i_{1}\leq d_{1}-1\\
|\phi_{i_{2}}\rangle&=|0i_{2}0\cdots0\rangle_{A_{1}A_{2}\cdots A_{n}}+\omega_{n-1}|00i_{2}0\cdots0\rangle_{A_{1}A_{2}\cdots A_{n}}+\cdots+\omega_{n-1}^{n-2}|0\cdots0i_{2}\rangle_{A_{1}A_{2}\cdots A_{n}},~d_{1}\leq i_{2}\leq d_{2}-1\\
|\phi_{i_{3}}\rangle&=|00i_{3}0\cdots0\rangle_{A_{1}A_{2}\cdots A_{n}}+\omega_{n-2}|000i_{3}0\cdots0\rangle_{A_{1}A_{2}\cdots A_{n}}+\cdots+\omega_{n-2}^{n-3}|0\cdots0i_{3}\rangle_{A_{1}A_{2}\cdots A_{n}},~d_{2}\leq i_{3}\leq d_{3}-1\\
& ~~~~~~~~~~~~~~~~~~\cdots~ \cdots ~\cdots~~~~~~~~~\\
|\phi_{i_{n-1}}\rangle&=|00\cdots0i_{n-1}0\rangle_{A_{1}A_{2}\cdots A_{n}}-|00\cdots0i_{n-1}\rangle_{A_{1}A_{2}\cdots A_{n}},~d_{n-2}\leq i_{n-1}\leq d_{n-1}-1\\
|\phi_{i_{n}}\rangle&=|00\cdots0i_{n}\rangle_{A_{1}A_{2}\cdots A_{n}}-|21\cdots1 (i_{n}-1)\rangle_{A_{1}A_{2}\cdots A_{n}},~d_{n-1}\leq i_{n}\leq d_{n}-1, \\
|S\rangle & =|0+\cdots+(d_{1}-1)\rangle_{A_{1}}|0+\cdots+(d_{2}-1)\rangle_{A_{2}}\cdots|0+\cdots+(d_{n}-1)\rangle_{A_{n}}.\\
\end{aligned}
\end{equation}
\end{widetext}

\end{theorem}

\begin{proof} First, we show that each of the first $(n-1)$ parties could only start with a trivial orthogonality-preserving measurement.
	 Suppose that the $k$th ($1\leq k\leq n-1$) party     starts with the first orthogonality-preserving measurement whose elements are  represented as $E_k=(m^k_{a,b})_{a,b\in\mathbb{Z}_{d_k}}.$  Then for  each pair $|\psi\rangle,|\phi\rangle\in \mathcal{S}$  with $|\psi\rangle\neq |\phi\rangle$, we have
	\begin{equation}\label{eq:d1d2d3dk}
		\langle \psi| I_1\otimes \cdots \otimes E_k\otimes \cdots \otimes I_n|\phi\rangle=0.
	\end{equation}
	
	Now we consider Eq.~\eqref{eq:d1d2d3dk} for  the states $|\phi_{i}\rangle$ and $|\phi_{j}\rangle$  for $0\leq i\neq j\leq d_k-1.$
	Note that the  parties except the $k$th of each term of  $|\phi_{i}\rangle$ and $|\phi_{j}\rangle$ are orthogonal except the terms $|0\cdots 0i0\cdots 0\rangle$ (corresponding to $|\psi_i\rangle$) and  $|0\cdots 0j0\cdots 0\rangle$ (corresponding to $|\psi_j\rangle$) where $i,j$ are in the $k$th position.
	Therefore, by Lemma \ref{lem:zero}, we could obtain that  $m^k_{i,j}=m^k_{j,i}=0.$  Therefore,   the off-diagonal entries of  $E_k$ are all zeros.

	For $1\leq i \leq d_k-1$, considering Eq.~\eqref{eq:d1d2d3dk} for the states $|S\rangle$ and $|\phi_{i}\rangle$, we   get $ m^k_{i,i}=m^k_{0,0}$   by Lemma~\ref{lem:Dia}.
	Therefore, $E_{k}\propto \mathbb {I}$. So the $k$th party cannot start with a nontrivial  orthogonality-preserving measurement.
	
	Let us consider the last party, i.e., the $n$th party.
	Suppose that  the $n$th party  starts with the first orthogonality-preserving measurement whose elements are  represented as $E_n=(m^n_{a,b})_{a,b\in\mathbb{Z}_{d_n}}.$  Then for  each pair $|\psi\rangle,|\phi\rangle\in \mathcal{S}$  with $|\psi\rangle\neq |\phi\rangle$, we have
	\begin{equation}\label{eq:d1d2d3dn}
		\langle \psi| I_1\otimes I_2\otimes\cdots \otimes I_{n-1} \otimes  E_n|\phi\rangle=0.
	\end{equation}

Considering Eq.~\eqref{eq:d1d2d3dn}  for the pair $|\phi_i\rangle$ and $|\phi_j\rangle$ where $0\leq i\neq j\leq d_n-1$, by  Lemma~\ref{lem:zero},  we could obtain that the off diagonal  entries $m^n_{i,j}$ of the matrix $E_{n}$ are zero except $d_{n-1}\leq i\neq j\leq d_n-1$ (see  Table \ref{5_result}). Now we only consider the remaining off diagonal entries of the matrix $E_{n}$.

For $d_{n-1}\leq i_n<i_n' \leq d_n-1, $ we consider Eq.~\eqref{eq:d1d2d3dn}  for the pair $|\phi_{i_n}\rangle$ and $|\phi_{i_n'}\rangle.$  That is,  $ [\langle 00\cdots0i_{n}|-\langle 21\cdots1(i_{n}-1)|]I_{1}\otimes I_{2} \otimes\cdots \otimes E_{n}[|00\cdots0i_{n}'\rangle -|21\cdots1(i_{n}'-1)\rangle]=0$
from which we deduce that
$m_{i_{n},i_{n}'}^{n}=-m_{i_{n}-1,i_{n}'-1}^{n}$.
Therefore, we have
$$m^n_{i_{n},i_{n}'}= (-1)^{i_{n}-d_{n-1}+1}m^n_{d_{n-1}-1,(i_{n}'-i_{n}+d_{n-1}-1)}=0$$
where  the last equality  has been deduced previously.
Thus we get $m^n_{i_{n},i_{n}'}=m^n_{i_{n}',i_{n}}=0.$  Hence we have  that the off-diagonal entries of the matrix $E_{n}$ are zeros.

By Lemma~\ref{lem:Dia}, all diagonal entries of the matrix $E_{n}$ are equal from Table~\ref{6_result}.

\begin{table}[htbp]
    \newcommand{\tabincell}[2]{\begin{tabular}{@{}#1@{}}#2\end{tabular}}
    \centering
    \caption{\label{6_result} ~Diagonal entries of~$E_{n}=(m_{a,b}^{n})_{a,b\in \mathbb Z_{d_{k}}}$.}
    \begin{tabular}{ccc}
        \toprule
        \hline
        \hline
        \specialrule{0em}{1.5pt}{1.5pt}
        A pair of states~&~Diagonal entries~&~Value Range\\
        \specialrule{0em}{1.5pt}{1.5pt}
        \midrule
        \hline
        \specialrule{0em}{1.5pt}{1.5pt}
        \tabincell{c}{$|S\rangle$, $|\phi_{i_{1}}\rangle$}&\tabincell{c}{$m_{0,0}^{n}= m_{i_{1},i_{1}}^{n}$}&\tabincell{c}{$1\leq i_{1}\leq d_{1}-1$  }\\
        \specialrule{0em}{3.5pt}{3.5pt}
        \tabincell{c}{$|S\rangle$, $|\phi_{i_{2}}\rangle$}&\tabincell{c}{$m_{0,0}^{n}= m_{i_{2},i_{2}}^{n}$}&\tabincell{c}{$d_{1}\leq i_{2}\leq d_{2}-1$  }\\

        \tabincell{c}{\vdots}&\tabincell{c}{\vdots}&\tabincell{c}{\vdots} \\

        \tabincell{c}{$|S\rangle$, $|\phi_{i_{n-1}}\rangle$}&\tabincell{c}{$m_{0,0}^{n}= m_{i_{n-1},i_{n-1}}^{n}$}&\tabincell{c}{$d_{n-2}\leq i_{n-1}\leq d_{n-1}-1$ }\\
        \specialrule{0em}{3.5pt}{3.5pt}
        \tabincell{c}{$|S\rangle$, $|\phi_{i_{n}}\rangle$}&\tabincell{c}{$m_{i_{n},i_{n}}^{n}= m_{i_{n}-1,i_{n}-1}^{n}$}&\tabincell{c}{$d_{n-1}\leq i_{n}\leq d_{n}-1$} \\
        \bottomrule
        \specialrule{0em}{1.5pt}{1.5pt}
        \hline
        \hline
    \end{tabular}
\end{table}

Therefore, all parties can only start with a trivial orthogonality preserving measurement. The set of $d_{n}+1$ orthogonal states is locally stable.
\end{proof}

 Moreover, we put forward a new construction of the locally stable set in $\otimes^{n}_{i=1}\mathbb{C}^{d_{i}}$ ($3\leq d_{1}\leq d_{2}\leq\cdots\leq d_{n}, n\geq 3$), which is composed of genuine entangled states apart from one full product state and also reach the minimum cardinality of the locally stable set proposed in Ref. \cite{Li2022}; see Appendix \ref{sec:AppA} for the details.

Many efforts have been made to reduce the cardinality of locally indistinguishable sets. Here we list the cardinalities of locally indistinguishable sets that have been known before (see Table~\ref{members}).  As locally stable sets are always locally indistinguishable,  there exists some locally indistinguishable sets with cardinality   $d_{n}+1$ in  $\mathbb{C}^{d_{1}}\otimes \mathbb{C}^{d_{2}}\otimes \cdots\otimes \mathbb{C}^{d_{n}}$ (where we assume $d_1\leq\cdots \leq d_n$). Thus our work has made a significant improvement towards addressing this issue.

\begin{table}[htbp]
	\newcommand{\tabincell}[2]{\begin{tabular}{@{}#1@{}}#2\end{tabular}}
	\centering
	\caption{\label{members}Incomplete list of the cardinalities of locally indistinguishable sets that are known before. }
	\begin{tabular}{ccc}
		\toprule
		\hline
		\hline
		\specialrule{0em}{1.5pt}{1.5pt}
		~reference~~&~~~system~~~&~~~cardinality~~~\\
		\specialrule{0em}{1.5pt}{1.5pt}
		\midrule
		\hline
		\specialrule{0em}{1.5pt}{1.5pt}
		\tabincell{c}{\cite{Yu2015}}&\tabincell{c}{$\mathbb{C}^{d}\otimes \mathbb{C}^{d}$}&\tabincell{c}{$2d-1$} \\
		\specialrule{0em}{3.5pt}{3.5pt}
		\tabincell{c}{\cite{Halder2018}}&\tabincell{c}{$(\mathbb{C}^{d})^{\otimes n}$}&~~\tabincell{c}{$2n(d+1)$} \\
		\specialrule{0em}{3.5pt}{3.5pt}
		\tabincell{c}{\cite{Zhen2022}}&\tabincell{c}{$(\mathbb{C}^{d})^{\otimes n}$}&~~\tabincell{c}{$n(d-1)+1$} \\
		\specialrule{0em}{3.5pt}{3.5pt}
		\tabincell{c}{\cite{Wangyl2015}}&\tabincell{c}{$\mathbb{C}^{m}\otimes \mathbb{C}^{n}$}&\tabincell{c}{$3(m+n)-9$} \\
		\specialrule{0em}{3.5pt}{3.5pt}
		\tabincell{c}{\cite{Zhangzc2016}}&\tabincell{c}{$\mathbb{C}^{m}\otimes \mathbb{C}^{n}$}&\tabincell{c}{$2n-1$} \\
		\specialrule{0em}{3.5pt}{3.5pt}
		\tabincell{c}{\cite{Xu2021}}&\tabincell{c}{$\mathbb{C}^{m}\otimes \mathbb{C}^{n}$}&\tabincell{c}{$2(m+n)-4$} \\
		\specialrule{0em}{3.5pt}{3.5pt}
		\tabincell{c}{\cite{Wangyl2017}}&\tabincell{c}{$\mathbb{C}^{n_1}\otimes \mathbb{C}^{n_2}\otimes \mathbb{C}^{n_3}$}&\tabincell{c}{$2(n_2+n_3)-3$} \\
		\specialrule{0em}{3.5pt}{3.5pt}
		\tabincell{c}{\cite{Jiang2020}}&\tabincell{c}{$\mathbb{C}^{d_{1}}\otimes \mathbb{C}^{d_{2}}\otimes \cdots\otimes \mathbb{C}^{d_{n}}$}&~~\tabincell{c}{$\sum_{i=1}^{n}(2d_{i}-3)+1$} \\		
		\specialrule{0em}{3.5pt}{3.5pt}
		\tabincell{c}{\cite{Zhen2022}}&\tabincell{c}{$\mathbb{C}^{d_{1}}\otimes \mathbb{C}^{d_{2}}\otimes \cdots\otimes \mathbb{C}^{d_{n}}$}&~~\tabincell{c}{$\sum\limits_{i=2}^{n-1}d_{i}+2d_{n}-n+1$} \\
		\bottomrule
		\specialrule{0em}{1.5pt}{1.5pt}
		\hline
		\hline
	\end{tabular}
\end{table}

\section{Conclusion}
We studied the construction of locally stable sets   for the given multipartite systems.
 It is interesting to note that the structures reach the lower bound of the cardinality on the locally stable sets. In fact, we presented the constructions of locally stable sets with minimum cardinality in bipartite quantum systems $\mathbb{C}^{d}\otimes \mathbb{C}^{d}$ and $\mathbb{C}^{d_{1}}\otimes \mathbb{C}^{d_{2}}$. Then we presented a construction of $d+1$ orthogonal states in $(\mathbb{C}^{d})^{\otimes n}$ and proved that the set is locally stable. Furthermore, we generalized our construction to more general cases and put forward two structures of $d_{n}+1$ orthogonal states for arbitrary multipartite quantum systems. Our results give a complete answer to the open problem raised in Ref. \cite{Li2022}. Moreover, all of our constructed locally stable  sets are
optimal in the sense that  removing any state from this set makes it impossible to achieve local stability again.

Here we have considered the constructions of the smallest locally stable sets by utilizing entangled states and a stopper state. However, there are two very important questions that deserve further research. Can we construct  the strongest nonlocal sets that  reach the corresponding lower bound? How do we quantify the strength of quantum nonlocality?

\section*{Acknowledgments}

This work is supported by the Natural Science Foundation of Hebei Province (Grant No. F2021205001), NSFC (Grants No. 62272208, No. 11871019, and No. 12005092).

\appendix

\section{PROOFS OF LEMMA \ref{lem:zero} AND LEMMA \ref{lem:Dia}}\label{append:A}

\noindent {\bf Proof of Lemma \ref{lem:zero}.}	
	Substituting the expressions
	$$
	\begin{array}{rcl}
		|\phi_{i}\rangle&=&\sum\limits^{p_{i}-1}_{t=0}\omega_{p_{i}}^{t}|i_{1}^{t}\rangle_{A_{1}}|i_{2}^{t}\rangle_{A_{2}}\cdots|i_{n}^{t}\rangle_{A_{n}},\\
		|\phi_{j}\rangle&=&\sum\limits^{p_{j}-1}_{s=0}\omega_{p_{j}}^{s}|j_{1}^{s}\rangle_{A_{1}}|j_{2}^{s}\rangle_{A_{2}}\cdots|j_{n}^{s}\rangle_{A_{n}}
	\end{array}$$ into
	\begin{equation}\notag
		\langle\phi_{i}|I_{1}\otimes \cdots\otimes E_{k}\otimes\cdots\otimes I_{n}|\phi_{j}\rangle =0,
	\end{equation}
 we obtain
	\begin{equation}\notag
		\begin{aligned}
			(\sum\limits^{p_{i}-1}_{t=0}\omega_{p_{i}}^{-t}\langle i_{1}^{t}|\langle i_{2}^{t}|\cdots \langle i_{n}^{t}|) I_{1}\otimes \cdots\otimes E_{k}\otimes\cdots \\ \otimes I_{n}(\sum\limits^{p_{j}-1}_{s=0}\omega_{p_{j}}^{s}|j_{1}^{s}\rangle|j_{2}^{s}\rangle\cdots|j_{n}^{s}\rangle) =0.
		\end{aligned}
	\end{equation}
	Further,
	\begin{equation}\notag
		\sum\limits^{p_{j}-1}_{s=0}\sum\limits^{p_{i}-1}_{t=0}\omega_{p_{i}}^{-t}\omega_{p_{j}}^{s} \langle i_{k}^{t}|E_{k}|j_{k}^{s}\rangle_{A_k} \prod_{\ell\neq k} \langle i_\ell^{t}| j_\ell^{s}\rangle_{A_\ell} =0.
	\end{equation}
	Since there is only one pair $(t_0,s_0) \in \mathbb{Z}_{p_i} \times \mathbb{Z}_{p_j}$ such that $\prod_{\ell\neq k} \langle i_\ell^{t_0}| j_\ell^{s_0}\rangle_{A_\ell}\neq 0$,  therefore we can get $\omega_{p_{i}}^{-t_{0}}\omega_{p_{j}}^{s_{0}}  \langle i_{k}^{t_{0}}|E_{k}|j_{k}^{s_{0}}\rangle_{A_k} \prod_{\ell\neq k} \langle i_\ell^{t_0}| j_\ell^{s_0}\rangle_{A_\ell} =0.$ Then $\langle i_{k}^{t_{0}}|E_{k}|j_{k}^{s_{0}}\rangle=0$, which means that ~$m^k_{i^{t_{0}}_{k},j^{s_{0}}_{k}}=0.$
\qed

\vskip 20pt

\noindent {\bf Proof of Lemma \ref{lem:Dia}.}
	 	Substituting the expressions
	 $$
	 \begin{array}{rcl}
	 	|S\rangle&=&\otimes^{n}_{k=1}(\sum\limits_{i_{k} \in \mathbb{Z}_{d_{k}}}|i_{k}\rangle_{A_{k}}),\\
	 	|\phi_{i}\rangle&=&\sum\limits^{p -1}_{t=0}\omega_{p }^{t}|i_{1}^{t}\rangle_{A_{1}}|i_{2}^{t}\rangle_{A_{2}}\cdots|i_{n}^{t}\rangle_{A_{n}}
	 \end{array}$$ into
	\begin{equation}\notag
		\langle S|I_{1}\otimes \cdots\otimes E_{k}\otimes\cdots\otimes I_{n}|\phi_{i}\rangle =0,
	\end{equation}
we obtain
	\begin{equation}\notag
		\begin{aligned}
			(\otimes^{n}_{k=1}(\sum\limits_{i_{k} \in Z_{d_{k}}}\langle i_{k}|) )I_{1}\otimes \cdots\otimes E_{k}\otimes\cdots \\ \otimes I_{n}(\sum\limits^{p-1}_{t=0}\omega_{p}^{t}|i_{1}^{t}\rangle|i_{2}^{t}\rangle\cdots|i_{n}^{t}\rangle)=0.
		\end{aligned}
	\end{equation}
	Further,
	\begin{equation}\notag
		\begin{aligned}
			\sum\limits^{p-1}_{t=0}\omega_{p}^{t}\langle 0+\cdots+(d_{1}-1)|i_{1}^{t}\rangle \cdots \langle 0+\cdots+(d_{k}-1)|E_{k}|i_{k}^{t}\rangle\\ \cdots \langle 0+\cdots+(d_{n}-1)|i_{n}^{t}\rangle =0.
		\end{aligned}
	\end{equation}
	Moreover,
	\begin{equation}\notag
		\sum\limits^{p-1}_{t=0}\omega_{p}^{t}\langle 0+\cdots+(d_{k}-1)|E_{k}|i_{k}^{t}\rangle=0.
	\end{equation}
	Since all $m_{a,b}^{k}=0$ with $0\leq a\neq b\leq d_{k}-1$, this means that
	\begin{equation}\notag
		m^k_{i^{0}_{k},i^{0}_{k}}+\omega_{p}m^k_{i^{1}_{k},i^{1}_{k}}+\omega^{2}_{p}m^k_{i^{2}_{k},i^{2}_{k}}+\cdots+\omega^{p-1}_{p}m^k_{i^{p-1}_{k},i^{p-1}_{k}}=0.
	\end{equation}
	
	If there exist only two different values $i^{t_{0}}_{k}$ and $i^{t_{1}}_{k}$ for $i^{0}_{k}, i^{1}_{k},\cdots, i^{p-1}_{k}$, this means that $p$ elements are divided into two groups. There may be ~$p-1$ elements equal, ~$p-2$ elements that are equal and other ~$2$ elements are equal, ~$p-3$ elements that are equal and the remaining ~$3$ elements are equal, etc. Here we only consider the following two cases; the others  can  be proved in a similar way.
	
	(1) Suppose ~$i^{0}_{k}=i^{t_{0}}_{k}$,~$i^{1}_{k}=\cdots=i^{p-1}_{k}=i^{t_{1}}_{k}$, then
	\begin{equation}\notag
		m^k_{i^{t_{0}}_{k},i^{t_{0}}_{k}}=-(\omega_{p}+\omega^{2}_{p}+\cdots+\omega^{p-1}_{p})m^k_{i^{t_{1}}_{k},i^{t_{1}}_{k}}.
	\end{equation}

	(2) Suppose ~$i^{0}_{k}=i^{1}_{k}=i^{t_{0}}_{k}$,~$i^{2}_{k}=\cdots=i^{p-1}_{k}=i^{t_{1}}_{k}$, then
	\begin{equation}\notag
		(1+\omega_{p})m^k_{i^{t_{0}}_{k},i^{t_{0}}_{k}}=-(\omega^{2}_{p}+\cdots+\omega^{p-1}_{p})m^k_{i^{t_{1}}_{k},i^{t_{1}}_{k}}.
	\end{equation}
	We know that ~$1+\omega_{p}+\omega^{2}_{p}+\cdots+\omega^{p-1}_{p}=0$; hence $m^k_{i^{t_{0}}_{k},i^{t_{0}}_{k}}=m^k_{i^{t_{1}}_{k},i^{t_{1}}_{k}}.$\qed

\section{ANOTHER STRUCTURE IN $\otimes^{n}_{i=1}\mathbb{C}^{d_{i}}$}
\label{sec:AppA}

\begin{theorem}
\label{the:dn'} The following set  $\mathcal{S}$ of $d_{n}$ orthogonal genuine entangled states and one full product state are locally stable in $\otimes^{n}_{i=1}\mathbb{C}^{d_{i}}$ for $3\leq d_{1}\leq d_{2}\leq\cdots\leq d_{n}$, $n\geq 3$:
\begin{widetext}
\begin{equation}\label{eq:genuine}
\begin{aligned}
|\phi_{0}\rangle &=|00\cdots00\rangle_{A_{1}A_{2}\cdots A_{n}}-|11\cdots11\rangle_{A_{1}A_{2}\cdots A_{n}}, \\
|\phi_{i_{1}}\rangle&=|i_{1}0\cdots00\rangle_{A_{1}A_{2}\cdots A_{n}}+\omega_{n}|0i_{1}0\cdots00\rangle_{A_{1}A_{2}\cdots A_{n}}+\cdots+\omega_{n}^{n-2}|00\cdots0i_{1}0\rangle_{A_{1}A_{2}\cdots A_{n}}\\
&~~~~+\omega_{n}^{n-1}|00\cdots0i_{1}\rangle_{A_{1}A_{2}\cdots A_{n}},~1\leq i_{1}\leq d_{1}-1\\
|\phi_{i_{2}}\rangle&=|0i_{2}0\cdots00\rangle_{A_{1}A_{2}\cdots A_{n}}+\omega_{n}|00i_{2}0\cdots00\rangle_{A_{1}A_{2}\cdots A_{n}}+\cdots+\omega_{n}^{n-2}|00\cdots0i_{2}\rangle_{A_{1}A_{2}\cdots A_{n}}\\
&~~~~+\omega_{n}^{n-1}|10\cdots0i_{2}\rangle_{A_{1}A_{2}\cdots A_{n}},~d_{1}\leq i_{2}\leq d_{2}-1\\
|\phi_{i_{3}}\rangle&=|00i_{3}0\cdots00\rangle_{A_{1}A_{2}\cdots A_{n}}+\omega_{n-1}|000i_{3}0\cdots00\rangle_{A_{1}A_{2}\cdots A_{n}}+\cdots+\omega_{n-1}^{n-3}|00\cdots0i_{3}\rangle_{A_{1}A_{2}\cdots A_{n}}\\ &~~~~+\omega_{n-1}^{n-2}|110\cdots0i_{3}\rangle_{A_{1}A_{2}\cdots A_{n}},~d_{2}\leq i_{3}\leq d_{3}-1\\
|\phi_{i_{4}}\rangle&=|000i_{4}0\cdots0\rangle_{A_{1}A_{2}\cdots A_{n}}+\omega_{n-2}|0000i_{4}0\cdots0\rangle_{A_{1}A_{2}\cdots A_{n}}+\cdots+\omega_{n-2}^{n-4}|00\cdots0i_{4}\rangle_{A_{1}A_{2}\cdots A_{n}}\\
&~~~~+\omega_{n-2}^{n-3}|1110\cdots0i_{4}\rangle_{A_{1}A_{2}\cdots A_{n}},~d_{3}\leq i_{4}\leq d_{4}-1\\
& ~~~~~~~~~~~~~~~~~~\cdots~ \cdots ~\cdots~~~~~~~~~\\
|\phi_{i_{n-1}}\rangle&=|0\cdots0i_{n-1}0\rangle_{A_{1}A_{2}\cdots A_{n}}+\omega_{3}|0\cdots0i_{n-1}\rangle_{A_{1}A_{2}\cdots A_{n}}+\omega_{3}^{2}|1\cdots10i_{n-1}\rangle_{A_{1}A_{2}\cdots A_{n}},~d_{n-2}\leq i_{n-1}\leq d_{n-1}-1\\
|\phi_{i_{n}}\rangle&=|00\cdots0i_{n}\rangle_{A_{1}A_{2}\cdots A_{n}}-|21\cdots1(i_n-1)\rangle_{A_{1}A_{2}\cdots A_{n}},~d_{n-1}\leq i_{n}\leq d_{n}-1 \\
|S\rangle & =|0+\cdots+(d_{1}-1)\rangle_{A_{1}}|0+\cdots+(d_{2}-1)\rangle_{A_{2}}\cdots|0+\cdots+(d_{n}-1)\rangle_{A_{n}}.\\
\end{aligned}
\end{equation}
\end{widetext}

\end{theorem}

\begin{proof} Comparing with Eq.~\eqref{eq:mulpro}, in Eq.~\eqref{eq:genuine}, we only made some slight adjustments such that ~$|\phi_{i_{1}}\rangle\sim |\phi_{i_{n}}\rangle$ are genuinely entangled states. Thus we only need to consider some special entries of the matrix $E_{n}$.

For the states $|\phi_{i_{2}}\rangle$ and $|\phi_{i_{2}'}\rangle$, where $d_{1}\leq i_{2}\neq i_{2}'\leq d_{2}-1$, we have $(\langle 0i_{2}\cdots0|+\cdots+\omega_{n}^{2-n}\langle0\cdots0i_{2}|+\omega_{n}^{1-n}\langle10\cdots0i_{2}|)I_{1}\otimes I_{2}\otimes \cdots\otimes E_{n}(|0i_{2}'\cdots0\rangle+\cdots+\omega_{n}^{n-2}|00\cdots0i_{2}'\rangle+\omega_{n}^{n-1}|10\cdots0i_{2}'\rangle)=0$. Because of the fact that only $|00\cdots0i_{2}\rangle_{A_{1}\cdots A_{n}}$ and $|00\cdots0i_{2}'\rangle_{A_{1}\cdots A_{n}}$, $|100\cdots0i_{2}\rangle_{A_{1}\cdots A_{n}}$ and $|100\cdots0i_{2}'\rangle_{A_{1}\cdots A_{n}}$ are not orthogonal on $n-1$ subsystems except the $n$th subsystem, so $\langle i_{2}|E_{n}|i_{2}'\rangle+\langle i_{2}|E_{n}|i_{2}'\rangle=0$; thus $m_{i_{2},i_{2}'}^{n}=m_{i_{2}',i_{2}}^{n}=0$ for $d_{1}\leq i_{2}\neq i_{2}'\leq d_{2}-1$.

~~~~~~~~~~~~~~~~~~~~~~~~~~~~~~~~~~\vdots

Similarly, from the states $|\phi_{i_{n-1}}\rangle$ and $|\phi_{i_{n-1}'}\rangle$, we can get $m_{i_{n-1},i_{n-1}'}^{n}=m_{i_{n-1}',i_{n-1}}^{n}=0$ for $d_{n-2}\leq i_{n-1}\neq i_{n-1}'\leq d_{n-1}-1$.

From the states $|\phi_{1}\rangle$ and $|\phi_{i_{2}}\rangle$, we have $(\langle 10\cdots0|+\cdots+\omega_{n}^{2-n}\langle00\cdots010|+\omega_{n}^{1-n}\langle00\cdots01|)I_{1}\otimes I_{2}\otimes \cdots\otimes E_{n}(|0i_{2}0\cdots0\rangle+\cdots+\omega_{n}^{n-2}|00\cdots0i_{2}\rangle+\omega_{n}^{n-1}|10\cdots0i_{2}\rangle)=0$.  Because only $|10\cdots0\rangle_{A_{1}\cdots A_{n}}$ and $|10\cdots0i_{2}\rangle_{A_{1}\cdots A_{n}}$, $|00\cdots01\rangle_{A_{1}\cdots A_{n}}$ and $|00\cdots0i_{2}\rangle_{A_{1}\cdots A_{n}}$ are not orthogonal on $n-1$ subsystems except the $n$th subsystem, then $\omega_{n}^{n-1}\langle 0|E_{n}|i_{2}\rangle+\omega_{n}^{-1}\langle 1|E_{n}|i_{2}\rangle=0$; i.e., $\omega_{n}^{n-1}m_{0,i_{2}}^{n}+\omega_{n}^{-1}m_{1,i_{2}}^{n}=0$, since $m_{0,i_{2}}^{n}=m_{i_{2},0}^{n}=0$, and we can get $m_{1,i_{2}}^{n}=m_{i_{2},1}^{n}=0$ for $d_{1}\leq i_{2}\leq d_{2}-1$.

Therefore, all parties can only start with a trivial orthogonality preserving measurement. The set of $d_{n}+1$ orthogonal states is locally stable.
\end{proof}

\end{document}